\documentclass[times]{ettauth}

\usepackage{moreverb}
\usepackage{amsthm}
\usepackage{amsthm}
\usepackage[colorlinks,bookmarksopen,bookmarksnumbered,citecolor=red,urlcolor=red]{hyperref}
\usepackage{epsfig}
\usepackage{cite}
\usepackage{caption}
\newcommand\BibTeX{{\rmfamily B\kern-.05em \textsc{i\kern-.025em b}\kern-.08em
T\kern-.1667em\lower.7ex\hbox{E}\kern-.125emX}}

\def\volumeyear{2016}

\newtheorem{definition}{Definition}
\newtheorem{theorem}{Theorem}
\newtheorem{lemma}{Lemma}
\newtheorem{remark}{Remark}
\newtheorem{col}{Corollary}

\begin{document}

\runningheads{H.~G.~Bafghi, et al.}{On The Secrecy of the Cognitive Interference Channel with Partial Channel States}

\articletype{RESEARCH ARTICLE}

\title{On The Secrecy of the Cognitive Interference Channel with Partial Channel States}

\author{Hamid G.~Bafghi$^{\dag}$, Babak Seyfe$^{\dag}$, Mahtab Mirmohseni$^{\ddag}$, Mohammad Reza Aref$^{\ddag}$}
\address{$\dag$ Information Theoretic Learning System Lab. (ITLSL), Department of Electrical Engineering, Shahed University, Tehran, Iran.\\
$\ddag$ Information Systems and Security Lab. (ISSL), Department of Electrical Engineering, Sharif University of Technology, Tehran, Iran.\\
E-mails: \{ghanizade, seyfe\}@shahed.ac.ir, \{mirmohseni, aref\}@sharif.edu.}

\begin{abstract}
The secrecy problem in the state-dependent cognitive interference channel is considered in this paper.
In our model, there are a primary and a secondary (cognitive) transmitter-receiver pairs,
in which the cognitive transmitter has the message of the primary one as side information.
In addition, the channel is affected by a channel state sequence which is estimated partially
at the cognitive transmitter and the corresponding receiver, separately.
The cognitive transmitter should cooperate with the primary one, and it
wishes to keep its message secure at the primary receiver.
The achievable equivocation-rate regions for this
channel are derived using two approaches: the binning scheme coding, and superposition coding. Then the outer
bounds on the capacity are derived and the results are extended to the Gaussian examples.
\end{abstract}


\maketitle

\section{Introduction}\label{Int}
Interference channel, in which the intended signal for one receiver causes interference at the other receivers, is a basic
model to study the constraints on the practical communication networks~\cite{bibi2}.
The Cognitive Interference Channel
(CIC) is one case of the interference channels in which one of the transmitter-receiver pair, namely the
primary one, has the privileges to use the channel~\cite{bibi9,bibi82}.
The secondary transmitter-receiver pair, i.~e., the cognitive
one, uses the channel without causing problem for the primary one. In one approach, the cognitive transmitter
cooperates with the primary party by spending the cognition cost~\cite{bibi83}.
Although the capacity of this channel remains an open problem in general case,
many works studied the achievable rate region for this channel
\cite{bibi12, bibi13, bibi71, bibi72}. Under degradedness condition,~\cite{bibi12} derived the capacity for the
CIC. The achievable rate of~\cite{bibi83} is improved by~\cite{bibi13, bibi71, bibi72, bibi138}.

The nature of the interference channel causes to leak the information to unintended destinations. In the information
theory literature, secure communication between the transmission parties was first studied in~\cite{bibi28} by Shannon.
Afterwards, Wyner introduced the wiretap channel to model the secrecy problem in the physical layer~\cite{bibi7}.
Furthermore, he proposed the~\textit{Random Coding} to keep the sent message away from the eavesdropper.
This coding scheme is based on the fact that a receiver cannot decode any information more than its
channel capacity with low-enough error probability.
Recently, there has been a significant interest in the
secrecy of multi-users systems~\cite{bibi79} with a particular
emphasis on the secrecy of the CIC~\cite{bibi9, bibi31, bibi87}.
The works~\cite{bibi9, bibi31, bibi87} derived some equivocation-rate regions
for the CIC to show the trade off between the achievable rate and the secrecy level in this channel.

Modeling a time-varying channel, whose instantaneous parameters depend on a random state sequence, is
introduced and studied by Shannon in his landmark paper~\cite{bibi39}. Moreover, the knowledge of the random state
sequence, i. e., the Channel State Information (CSI) is assumed to be available at the transmitter in~\cite{bibi39}.
There are considerable research interests in studying the effect of the CSI in various channel models (see~\cite{bibi69} and the references therein).
Specifically, the capacity of a discrete memoryless point to point channel with non-causal CSI available at the
Transmitter (CSIT) is derived by Gel'fand and Pinsker~\cite{bibi22}, and it is extended to the Gaussian channel in~\cite{bibi30}.
The CIC with CSI available at the cognitive transmitter is studied in~\cite{bibi82, bibi47} and the equivocation-rate region on this model is derived by~\cite{bibi70}. Moreover, some works consider the impact of partial channel state
information on the capacity and performance of the cognitive radio~\cite{bibi135,bibi136}.

In this paper, we study the CIC with Partial Channel State information (CIC-PCSI).
The partial CSIs are assumed to be known non-causally at the cognitive
transmitter and the corresponding receiver (see Figure~\ref{fig:1}). Here,
the cognitive transmitter should mitigate its
interference at the primary receiver.
Furthermore, it wishes to keep its message confidential with respect to
the primary receiver.

The CIC-PCSI model can be motivated by the wireless sensor network
application with different sensor types~\cite{bibi12},
in which one sensor has a better sensing capability than the other one.
The simpler sensor provides one event to its corresponding destination, but the more capable sensor
which can sense two events, cooperates with the simpler sensor.
Since the more capable sensor senses
a vital event, it wishes to keep its message confidential at the destination of the simpler sensor.
Moreover, the channel is affected by a channel state sequence which is estimated at the more capable sensor and
its destination, separately~\cite{bibi41}. These estimated observations of the channel state sequence are not equal
to each other in general case.

We study the different effects of the CSIT and CSIR on two coding schemes to achieve the equivocation-rate region.
For this aim, we use the~\textit{Binning scheme}~\cite{bibi11, bibi10, bibi13} and the~\textit{Superposition Coding}~\cite{bibi9, bibi82, bibi70}
in CIC-PCSI. In the binning scheme, the cognitive transmitter, after rate splitting,
bins its message against the code-book of the primary one. Then, it superimposes its message on the primary transmitter's message and the channel state sequence. In the superposition scheme, the cognitive transmitter superimposes its message on the primary transmitter's message and the channel state sequence. In each scheme, random coding is used to guaranty the secrecy condition for the cognitive transmitter's message~\cite{bibi7}.
Then, the outer bounds on the capacity of the CIC-PCSI are proposed.
Moreover, we extend the results of two cases, i.e., binning scheme and superposition coding, to the Gaussian model, and it is shown that the cognitive transmitter can choose the best coding scheme to maximize the achievable equivocation-rate region.
In comparison of our model with the different ones in~\cite{bibi82, bibi11, bibi10, bibi13},
we consider secrecy constraints in the CIC. Since we assume that the primary transmitter's message is fully known at the cognitive transmitter,
the secrecy issue is considered for the cognitive transmitter's message (see the similar model in~\cite{bibi9}).
Furthermore, in compare with the model of~\cite{bibi31} and~\cite{bibi87}, the CSI knowledge enhances the cognitive transmitter to improve
the equivocation-rate region.

The rest of the paper is as follows. In Section~\ref{S2},
the channel model and some preliminaries and the definitions are
explained. In Section~\ref{S3}, the main results on the achievable
equivocation-rate region using the binning scheme are
proposed. Furthermore, in this section we derive the proper outer
bounds on the capacity, and extend the results to the Gaussian
case. In Section~\ref{S4}, using the superposition coding, we
derive the equivocation-rate region and an
outer bound on the capacity of the channel. Then, we extend the
results to the Gaussian channel as an example.
Finally, the paper is concluded in Section~\ref{S5}.
The proofs of the theorems are relegated to the appendices.

\begin{figure}
\centering
\epsfig{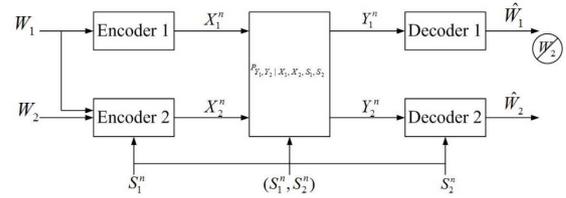}
\caption{The cognitive interference channel with two partial channel states information
available at the transmitter and the receiver, with a confidential message.}
\label{fig:1}
\end{figure}


\section{Channel Model and Preliminaries}\label{S2}
\subsection{The Notation}
First, we explain the notation. We use~$\mathcal{X}$ to denote a finite alphabet with cardinality~$|\mathcal{X}|$. $x^{n}=\{x_{1}, x_{2}, \ldots, x_{n}\}$ represents the members of~$\mathcal{X}^{n}$, in which the subscripted and superlative letters represent the components and the vectors, respectively. $x^{j}_{i}$ is used to indicate the vector~$(x_{i},\ldots, x_{j})$. For the random vectors and the random variables, which are denoted by uppercase letters, a similar convention is used.

\subsection{Channel Model}
Consider a memoryless stationary state-dependent interference channel with finite input alphabets~$\mathcal{X}_{1}$ and~$\mathcal{X}_{2}$, finite output alphabets~$\mathcal{Y}_{1}$ and~$\mathcal{Y}_{2}$, the channel states alphabets~$\mathcal{S}_{1}$,~$\mathcal{S}_{2}$ with distribution~$\mathcal{P}_{S_{1}}$,~$\mathcal{P}_{S_{2}}$ and a conditional probability distribution~$P_{Y_{1},Y_{2}|X_{1},X_{2},S_{1},S_{2}}$.
As shown in the Figure~\ref{fig:1}, the~$t$-th transmitter, where~$t=1,2$, wishes to transmit the message~$W_{t}$ which is uniformly distributed on the set~$\mathcal{W}_{t}\in \{1,\ldots M_{t}\}$. The message~$W_{1}$ is assumed to be known at both transmitters, but the message~$W_{2}$ is just known at the transmitter 2 (the cognitive transmitter). Furthermore, it is assumed that the channel is dependent on two channel states. One of these channel states, i.e.,~$S_{1}$, is assumed to be known non-causally at the cognitive transmitter. The other one, i.e.,~$S_{2}$, is assumed to be known non-causally at the cognitive destination. Thus, the cognitive party wishes to increase its achievable rate using this side information.

Given the inputs and the states, i.e., the~$n$-sequences~$X^{n}_{1}$,~$X^{n}_{2}$,~$S^{n}_{1}$,~$S^{n}_{2}$, the conditional distribution of the channel outputs~$n$-sequences~$Y^{n}_{1}$,~$Y^{n}_{2}$ take the product form as follows
\begin{eqnarray}\label{eqn1}
\!\!\!\!\!\!\!\!\!\!\!\!\!\!\!\! &&P_{Y^{n}_{1}, Y^{n}_{2}|X^{n}_{1}, X^{n}_{2}, S^{n}_{1}, S^{n}_{2}}(y^{n}_{1}, y^{n}_{2}|x^{n}_{1}, x^{n}_{2}, s^{n}_{1}, s^{n}_{2})=\nonumber\\
\!\!\!\!\!\!\!\!\!\!\!\!\!\!\!\! &&\prod_{i=1}^{n} P_{Y_{1}, Y_{2}|X_{1}, X_{2}, S_{1}, S_{2}}(y_{1,i}, y_{2,i}|x_{1,i}, x_{2,i}, s_{1,i}, s_{2,i}).
\end{eqnarray}

\subsection{Definitions}
An~$(M_1,M_2, n, P_e)$-code has two encoding-decoding functions and an error probability.
The encoding functions are defined as
\begin{eqnarray}{}\label{eqn2}
&&\varphi_{1,n}: \mathcal{W}_{1}\rightarrow  \mathcal{X}_{1}^{n},\nonumber\\
&&\varphi_{2,n}: \mathcal{W}_{1}\times \mathcal{W}_{2}\times \mathcal{S}^{n}_{1}\rightarrow \mathcal{X}_{2}^{n},
\end{eqnarray}
and the channel decoders are defined by the mappings
\begin{eqnarray}{}\label{eqn3}
&&\psi_{1,n}: \mathcal{Y}_{1}^n\rightarrow \mathcal{\hat{W}}_{1},\nonumber\\
&&\psi_{2,n}: \mathcal{Y}_{2}^n\times \mathcal{S}^{n}_{2}\rightarrow \mathcal{\hat{W}}_{2}.
\end{eqnarray}

The error probability~$P_{e}= \max (P_{e,1}, P_{e,2})$ is defined as
\begin{eqnarray}{}\label{eqn4}
P_{e,t}= \sum_{w_{1}, w_{2}}\frac{1}{M_{1} M_{2}}P[\hat{w}_{t}\neq w_{t}| w_{1}, w_{2} ~ \textmd{were sent}].
\end{eqnarray}

\begin{definition}
The secrecy level of the cognitive transmitter's message at the primary receiver (receiver~1) is
measured by normalized equivocation-rate which is defined as
\begin{eqnarray}{}\label{eqn5}
R_{e_2}^{(n)}= \frac{1}{n} H(W_2|Y_1^n),
\end{eqnarray}
which is known as the ``weak secrecy condition''~\cite{bibi79}.
\end{definition}

\begin{definition}
The equivocation-rate-triple~$(R_{1},R_{2}, R_{e_2})$ is an achievable region if for any~$\epsilon_{n}> 0$ there exists an~$(M_{1}, M_{2}, n, P_{e})$ code such that~$M_i \geq 2^{nR_i}$,~$i=1,2$ for which we have~$P_{e}\leq\epsilon_n$, and
\begin{eqnarray}{}\label{eqn6}
0 \leq R_{e_{2}}\leq \liminf_{n\rightarrow \infty} R_{e_{2}}^{(n)}.
\end{eqnarray}
\end{definition}

\begin{definition}
The capacity region is the closure of the set of all achievable equivocation-rate regions.
\end{definition}

\subsection{Encoding Schemes}

Now, we discuss the rate achieving encoding schemes we will use in the CIC problem. First, consider a point-to-point
state-dependent communication system in which the CSI is known non-causally at the transmitter. Assume that the channel state
sequence~$S$ plays the role of the interference signal which can be considered as a code-book with rate~$R_S = I(Y ; S)$.
The transmitter wishes to transmit the message~$W$ at the rate~$R$ through the channel. There are two coding schemes
to achieve the rate region:~\textit{Superposition Coding} (SPC) and~\textit{Gel'fand-Pinsker Coding} (GPC); depending on the
interference's rate~$R_S$, either one may be chosen. When~$R_S$ is small, we can improve the achievable rate using the
SPC. For higher~$R_S$, we can achieve the rate using the classical GPC. The following lemma expresses the result
using these two coding schemes\cite[Lemma~1]{bibi13}. This lemma is used to derive the achievable rate regions for the
CIC-PCSI in the next sections.

\begin{lemma}{\cite[Lemma 1]{bibi13}}\label{lem1}
The following rate region is achievable for a point to point communication system with non-causal CSIT
\begin{eqnarray}{}\label{eqn7}
\!\!\!\! && R \leq \max_{P_{U|S}, f(.)}\min\{ I(X; Y |S),\nonumber\\
\!\!\!\!\!\!\!\! &&\max\{I(U, S; Y ) - R_S, I(U; Y ) - I(U; S)\}\}.
\end{eqnarray}
\end{lemma}
\begin{proof}[Outline of the proof]
For the case~$I(S;U,Y) \leq R_S \leq H(S)$, the binning scheme achieves the rate given by the
second term of~\eqref{eqn7}. For~$R_S \leq I(S;U, Y )$, SPC achieves the rate region by the first term in~\eqref{eqn7}.
For more details we refer to~\cite[Lemma 1]{bibi13}.
\end{proof}
\section{Using the Binning Scheme}\label{S3}
In this section we derive the achievable equivocation-rate region for the CIC-PCSI, shown in Figure~\ref{fig:1}, using the binning scheme.
Then, two outer bounds on the capacity are proposed, and the results are extended to the Gaussian channel as special case.

\subsection{An inner bound}
To derive an achievable rate region for this channel, we use the rate splitting as follows.
\begin{eqnarray}{}\label{eqn8}
R_{1}&=&R_{1a}+ R_{1b},\\\label{eqn9}
R_{2}&=&R_{2a}+ R_{2b},
\end{eqnarray}
for non-negative rates~$R_{1a}, R_{1b}, R_{2a}$ and~$R_{2b}$.
Transmitter~1, encodes the message~$W_1$ and uses the SPC with two code-books~$X^n_{1a}$ and~$X^n_{1b}$. Transmitter~2, by access to
the message~$W_1$ and the channel state~$S^n_1$ uses the SPC with two code-books~$X^n_{1a}$ and~$X^n_{1b}$. Then, it splits the
message~$W_2$ and uses GPC against~$X^n_{1a}$,~$X^n_{1b}$,~$S^n_1$ in two steps to create~$X^n_2$. In the first step, transmitter~2 uses
binning against~$X^n_{1a}$,~$X^n_{1b}$,~$S^n_1$ to create~$U^n$ of rate~$R_{2b}$. In the second step, it uses binning
against~$X^n_{1a}$,~$X^n_{1b}$ and~$S^n_1$ conditioned on~$U^n$ to create~$V^n$ of rate~$R_{2a}$. Finally,
it uses~$U^n, V^n, X^n_{1a}, X^n_{1b}, S^n_1$ to construct~$X^n_{2}$.
Based on this encoding scheme, we have the following result on the achievable equivocation-rate region.

\begin{theorem} [Achievable equivocation-rate region]\label{thm1}
The set of equivocation-rates $(R_{1a},R_{1b},R_{2a},R_{2b},R_{e_2})$ is achievable if it satisfies
\begin{eqnarray}{}\label{eqn10}
&&R_1 \leq I(X_1; Y_1,U|Q),\\\label{eqn11}
&&R_{1b} \leq I(X_{1b}; Y_1,U |X_{1a}, Q),\\\label{eqn12}
&&R_{2a} \leq I(V ; Y_2, S_2|U, Q) - I(V ;X_1, S_1|U, Q),\\\label{eqn13}
&&R_2 \leq I(V,U; Y_2, S_2|Q) - I(V,U;X_1, S_1|Q),\\\label{eqn14}
&&R_1 + R_{2b} \leq I(X_1,U; Y_1|Q),\\\label{eqn15}
&&R_{1b} + R_{2b} \leq I(X_{1b},U; Y_2|X_{1a}, Q),\\\label{eqn16}
&&R_{e_2} \leq I(V ; Y_2, S_2, U | Q) - I(V, S_1; X_1, Y_1,U | Q),\nonumber\\
&&
\end{eqnarray}
for input distribution factors as
\begin{eqnarray}{}\label{eqn17}
&&p(q)p(x_{1a}, x_{1b}, u, v, x_1, x_2, s_1, s_2|q)\nonumber\\
&&\times p(y_1, y_2|x_1, x_2, s_1, s_2),
\end{eqnarray}
in which the right-hand-sides (r.h.s.) of the equations~\eqref{eqn10}--\eqref{eqn16} are non-negative and~$Q$ is a time-sharing random variable.
\end{theorem}
\begin{proof}
See Appendix~\ref{A}.
\end{proof}
Using the Fourier-Motzkin elimination~\cite{bibi69}, the following explicit description of the region is derived.
\begin{col}
The set of equivocation-rates~$(R_1,R_2,R_{e_2})$ is achievable if it satisfies
\begin{eqnarray}{}\label{eqn18}
\!\!\!\!\!\!\!\!\!\!\!\!\!\!\!\!\!\!\!&&R_1 \leq \min\{I(X_1; Y_1,U |Q), I(X_1,U; Y_1|Q)\},\\\label{eqn19}
\!\!\!\!\!\!\!\!\!\!\!\!\!\!\!\!\!\!\!&&R_2 \leq I(V,U; Y_2, S_2|Q) - I(V,U;X_1, S_1|Q),\\\label{eqn20}
\!\!\!\!\!\!\!\!\!\!\!\!\!\!\!\!\!\!\!&&R_1 + R_2 \leq I(V ; Y_2, S_2|U, Q) - I(V ;X_1, S_1|U, Q)\nonumber\\
\!\!\!\!\!\!\!\!\!\!\!\!\!\!\!\!\!\!\!&&+ I(X_1,U; Y_1|Q),\\\label{eqn21}
\!\!\!\!\!\!\!\!\!\!\!\!\!\!\!\!\!\!\!&&R_{e_2} \leq I(V ; Y_2, S_2,U | Q) - I(V, S_1;X_1, Y_1,U | Q),
\end{eqnarray}
for input distribution factors as~\eqref{eqn17}.
\end{col}

\begin{remark}
Theorem~1 without secrecy aspect and by substituting~$S_{1}=S_{2}=\emptyset$, is reduced to the result of~\cite[Theorem 1]{bibi13} for the CIC. Moreover, the equivocation-rate~\eqref{eqn16}, by substituting~$S_{1}=S_{2}=\emptyset$, is reduced to equivocation-rate of~\cite[Theorem 1]{bibi31}.
It means that Theorem~1 includes the results of~\cite{bibi13} and~\cite{bibi31}.
\end{remark}

\subsubsection{The Symmetric Channel State}
The special case~$S_1 = S_2 = S$ is of special interest. This case resembles the
secret-key agreement scenario~\cite{bibi69, bibi40}.
The equivocation-rate~\eqref{eqn16} in this case is reduced to the following theorem:
\begin{theorem}\label{thm2}
The secrecy-rate (SR) of the CIC, when the state sequence~$s^n$ is known at the transmitter and the receiver, is given by
\begin{eqnarray}{}\label{eqn22}
R_{e_2}^{SR} &\leq& I(V ; Y_2 | U, S) - I(V ;X_1, Y_1 | U, S) \nonumber\\
&&+ H(S | U,X_1, Y_1).
\end{eqnarray}
\end{theorem}
\begin{proof}
The achievability of~\eqref{eqn22} results from~\eqref{eqn16} as follows.
\begin{eqnarray}{}\nonumber
R_{e_2} &\leq& I(V ; Y_2, S,U) - I(V, S; X_1, Y_1,U)\nonumber\\
&=&    I(V ; Y_2, U | S)   - I(V ;X_1, Y_1,U | S) \nonumber\\
&&+ I(V ; S) - I(S;X_1, Y_1,U)\nonumber\\
&=&    I(V ; Y_2, U | S)   - I(V ;X_1, Y_1,U | S) \nonumber\\
&&+ H(S | U, X_1, Y_1) - H(S | V )\nonumber\\\label{eqn23}
&\leq& I(V ; Y_2,U | S)    - I(V ;X_1, Y_1,U | S) \nonumber\\
&&+ H(S | U, X_1, Y_1),
\end{eqnarray}
in which the last inequality follows from the non-negativity of the entropy function.
Note that~$V$ is an optimal choice. Therefore, selecting~$V=(V, S)$ leads to~$H(S|V)=0$, and
the bound in the last inequality will be tight.
An alternative proof can be derived directly from the secret-key agreement method taken in~\cite[Theorem 3]{bibi40}.
\end{proof}

\begin{remark}
The inner bound of Theorem~\ref{thm2} can be interpreted from the secret-key agreement point of view~\cite[Theorem 3]{bibi40}.
The term~$I(V ; Y_2 | U, S) - I(V ; X_1, Y_1 | U, S)$ is the rate of a multiplexed
CIC in which the cognitive transmitter and both the receivers (the primary and the secondary receivers),
have knowledge of~$s^n$ and the common message~$u^n$, non-causally. The second term~$H(S | U, X_1, Y_1)$
is the additional secret-key rate which can be produced by using
the fact that the channel state~$s^n$ is only known to the cognitive transmitter-receiver pair.
For more details on using the channel state as a shared secret-key between the transmitter-receiver pair, see~\cite{bibi40}.
\end{remark}

\subsection{Outer bounds}
The following theorems provide two outer bounds on the capacity region of the CIC-PCSI. In the first outer bound, we use the
usual approach taken in the previous work~\cite{bibi13, bibi31} based on the Fano's inequality.
In the second outer bound, we use the approach taken by~\cite{bibi71},
which only depends on the conditional marginal distributions
of the channel outputs given the inputs. This outer bound
does not include auxiliary random variables and every mutual
information term involves the inputs and outputs of the channel.
Therefore, the second outer bound
is looser than the first one, but can be more easily evaluated.

\begin{theorem}[Outer bound 1]\label{thm3}
The set of rates~$(R_{1},R_{2},R_{e_{2}})$ satisfying
\begin{eqnarray}{}\label{eqn24}
&&R_{1} \leq   \min\{I(U, V_1; Y_1), I(V_1; Y_1,U)\},\\\label{eqn25}
&&R_{2} \leq     I(U, V_2; Y_2|S_1, S_2),\\\label{eqn26}
&&R_1+R_2 \leq \min\Big\{I(V_2; Y_2|U, V_1, S_1, S_2) \nonumber\\
&&+ I(V_1,U; Y_1),I(V_2,U; Y_2|S_1, S_2) \nonumber\\
&&+ I(V_1; Y_1|U, V_2) \Big\},
\end{eqnarray}
\begin{eqnarray}\label{eqn27}
&&R_{e_2}\leq   \min \Big\{I(V_2; Y_2 | U) - I(V_2; Y_1 | U), \nonumber\\
&&I(V_2; Y_2 | V_1,U) - I(V_2; Y_1 | V_1, U)\Big\},
\end{eqnarray}
for input distribution that factors as
\begin{eqnarray}{}\label{eqn28}
&&p(s_1)p(s_2)p(v_1)p(v_2)p(u|v_1, v_2)p(x_1|v_1)\nonumber\\
&&p(x_2|u, v_1, v_2, s_1)p(y_1, y_2|x_1, x_2, s_1, s_2),
\end{eqnarray}
is an outer bound on the capacity of this channel.
\end{theorem}
\begin{proof}
The proof of Theorem~\ref{thm3} is relegated to Appendix~\ref{B}.
\end{proof}

\begin{theorem}[Outer bound 2]\label{thm4}
The set of rates~$(R_{1},R_{2},R_{e_{2}})$ satisfying
\begin{eqnarray}{}\label{eqn29}
&&R_{1} \leq   I(X_1, X_2; Y_1),\\\label{eqn30}
&&R_{2} \leq   I(X_2; Y_2|X_1, S_1, S_2),\\\label{eqn31}
&&R_1+R_2 \leq I(Y_1;X_1,X_2, S_1, S_2) \nonumber\\
&&+ I(X_2; Y_2|X_1, S_1, S_2,Y_1^{'} ),\\\label{eqn32}
&&R_{e_2}\leq   \min \Big\{I(X_2; Y_2) - I(X_2; Y_1), I(X_2; Y_2 | X_1) \nonumber\\
&&- I(X_2; Y_1 | X_1)\Big\},
\end{eqnarray}
for all distributions~$P_{X_1 , X_2}$ and~$P_{Y_2, Y_{1}^{'}| X_1, X_2, S_1, S_2}$, where~$Y_{1}^{'}$ has the same marginal distribution as~$Y_{1}$, i.~e.,~$P_{Y_{1}^{'}| X_1, X_2, S_1, S_2}=P_{Y_{1}| X_1, X_2, S_1, S_2}$, is an outer bound on the capacity of this channel.
\end{theorem}

\begin{proof}[Outline of the proof]
The rates~\eqref{eqn29}--\eqref{eqn31} are derived using the side information approach taken by~\cite{bibi71}.
The rate~\eqref{eqn32} is derived according to the previous rate~\eqref{eqn27} by substituting the
auxiliary random variables~$V_1$ and~$ V_2$ by~$ X_1$ and~$ X_2$, respectively.
This outer bound is looser than the one in Theorem~\ref{thm3}, but it does not include auxiliary random variables and
thus it can be more easily evaluated. The details on the proof are relegated to Appendix~\ref{C}.
\end{proof}

\subsection{The Gaussian example}
To clarify our results more perceptibly, consider the Gaussian CIC-PCSI.
The channel model is shown in Figure~\ref{fig:2}, and can be described as follows:
\begin{figure}
\centering
\epsfig{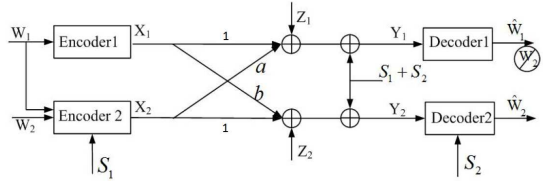}
\caption{The Gaussian cognitive interference channel with channel state
available at the transmitter and the receiver, with a confidential message.}
\label{fig:2}
\end{figure}
\begin{eqnarray}{}\nonumber\label{eqn37}
Y_{1}  &=&X_{1}+ aX_{2}+  S_1+ S_2 + Z_{1},\\
Y_{2}  &=&bX_{1}+ X_{2} + S_1+ S_2+ Z_{2},
\end{eqnarray}
where~$X_{i}$ and~$Y_{i}$ denotes the input and the output of the~$i$-th transmitter-receiver pair.~$Z_{i}\sim \mathcal{N}(0, 1)$
is Additive White Gaussian Noise (AWGN) at the~$i$-th receiver where~$i\in\{1,2\}$.~$S_i\sim \mathcal{N} (0,K_i)$
denotes the partial channel
state sequences which are known at the cognitive transmitter and the corresponding receiver, respectively.
The constants~$a$ and~$b$
are the real-valued channel gains in the interfering links and the average power constraint
is~$\frac{1}{n}\sum_{k=1}^{n}(X_{i, k}(t))^{2}\leq P_{i}$, $i\in\{1,2\}$.
In this model, for simplicity, we consider the partial channel state sequences to be additive and independent Gaussian
random variables.
This model can be motivated by the case in which two different interfering signals affect the channel,
and each one is estimated at one of the cognitive transmitter-receiver nodes.
Now, we consider the cases in which~$a \leq 1$ and~$a > 1$, separately.

\subsubsection{The case~$a \leq 1$}
This case is reported as the \textit{weak interference} case in the literature~\cite{bibi9, bibi82}. The capacity
region of the CIC in this case without CSI is determined by~\cite{bibi12, bibi80}, in which the cognitive encoder uses Dirty Paper
Coding (DPC) [19] for~$W_2$ against~$W_1$. Furthermore, using the SPC in the cognitive transmitter, the message~$W_1$ is
conveyed to receiver~2. In the weak interference case, receiver~2 does not suffer from the interference,
since, transmitter~2 uses DPC on~$X_2$ against~$X_1$ and known channel state. Moreover, the primary receiver
is not affected by the interfering signal~$X_2$ due to the weak interference. The following theorem describes the
achievable equivocation-rate region of the Gaussian CIC-PCSI in this case.

\begin{theorem}[Achievable equivocation-rate region]\label{thm5}
The set of rates~$(R_1, R_2,R_{e_2} )$ satisfying
\begin{eqnarray}{}\label{eqn38}
R_1   &\leq& \mathcal{C}\Big(\frac{P_1}{K_2+1}\Big),\\\label{eqn39}
R_{2} &\leq& \mathcal{C}((1-\rho^2)P_2),\\\label{eqn40}
R_{e_2}&\leq& \mathcal{C}((1-\rho^2)P_2)-\mathcal{C}((1-\rho^2)a^2P_2),
\end{eqnarray}
in which~$\mathcal{C}(x)=\frac{1}{2}(1+x)$ and~$\rho\in [0, 1]$, is an achievable equivocation-rate region of the Gaussian CIC-PCSI,
shown in Figure~\ref{fig:2} for the case~$a\leq 1$.
\end{theorem}
\begin{proof}
The proof is similar to the one presented in~\cite{bibi80} without secrecy and by substituting~$X_i\sim\mathcal{N}(0, P_i)$
for~$i\in\{1,2\}$ and~$E[X_1X_2]=\rho\sqrt{P_1P_2}$.
The channel state~$S_1$ is treated as interference by the cognitive
transmitter in DPC and does not affect the rate. On the other hand, the channel state~$S_2$, which is known
non-causally at the cognitive receiver, can be easily canceled out. Thus, these channel states do
not affect the rate~\eqref{eqn39}. The primary receiver~1 is affected by the channel state~$S_2$ as an additional interference,
but the channel state~$S_1$ is canceled out for this receiver by the cognitive transmitter's cooperation.
For more details on the proof see~\cite{bibi80}.
\end{proof}

\begin{remark}
The achievable equivocation-rate region for the Gaussian CIC-PCSI in Theorem~\ref{thm5} is maximized for $\rho=0$ since $a\leq1$. Thus in this case, the cognitive transmitter meets its capacity and the equivocation leads to~$\mathcal{C}(P_2)-\mathcal{C}(a^2P_2)$.
\end{remark}
\subsubsection{The case~$a > 1$}
In this case, which is known as the \textit{strong interference}~\cite{bibi9, bibi82},
the channel output at the cognitive receiver
is a degraded version of that at the primary one, thus there is no secrecy
in this condition, i.~e.,~$R_{e_2} = 0$.
In this case, receiver~1, having better observation of~$X_2$ than the cognitive receiver, can decode
the message of the cognitive transmitter without any penalty rate. The capacity of the CIC without channel state~\cite{bibi9, bibi80},
is a trivial outer bound on the capacity of the CIC-PCSI. This outer bound is presented in the following.

\begin{theorem}[Gaussian outer bound{\cite[Theorem 2]{bibi9}}\label{thm6}]
The set of rates~$(R_1,R_2)$ satisfying
\begin{eqnarray}{}\label{eqn41}
R_2   &\leq& \mathcal{C}((1-\rho^2)P_2),\\\label{eqn42}
R_1+R_2&\leq& \mathcal{C}(P_1+ a^2P_2+2a\rho\sqrt{P_1P_2}),\\\label{eqn43}
R_1+R_2&\leq& \mathcal{C}(b^2P_1+P_2+2b\rho\sqrt{P_1P_2}),
\end{eqnarray}
is an outer bound on the capacity of the Gaussian CIC-PCSI for the case~$a > 1$.
\end{theorem}

\section{Using the Superposition Coding}\label{S4}
The cognitive transmitter can superimpose part of its message on~$X^n_1$ instead of binning. Thus, it should split
its message as~$W_2 = W_{21} + W_{22}$, in which~$W_{21}$ is intended to both receivers and~$W_{22}$ is only decodable at the
cognitive receiver. Moreover, the cognitive transmitter uses GPC via three auxiliary random variables~$T$ ,~$U$ and~$V$
to reduce the channel state interference for~$W_1$,~$W_{21}$ and~$W_{22}$, respectively. In particular,~$T$ deals with state
interference for either receiver~1 or receiver~2 to decode~$W_1$;~$ U$ deals with state interference for either receiver~1
or receiver~2 to decode~$W_{21}$; and~$V$ deals with state interference for receiver~2 to decode~$W_{22}$. Now, we propose
the main results which are derived based on this scheme.
\subsection{An Inner Bound}
\begin{theorem}[Achievable equivocation-rate region]\label{thm7}
The set of rates~$ (R_1,R_{21},R_{22},R_{e_2} )$ is achievable if it satisfies
\begin{eqnarray}{}\label{eqn44}
\!\!\!\!\!\!\!\!\!\!\!\!&&R_1+R_{21}      \leq I(T,U,X_1; Y_1) - I(T,U; S_1|X_1),\\\label{eqn45}
\!\!\!\!\!\!\!\!\!\!\!\!&&R_{22}          \leq I(V ; Y_2, S_2|U,X_1, T) - I(V ; S_1|U,X_1, T),\\\label{eqn46}
\!\!\!\!\!\!\!\!\!\!\!\!&&R_{2}           \leq I(U, V ; Y_2, S_2|X_1, T) - I(U, V ; S_1|X_1, T),\\\label{eqn47}
\!\!\!\!\!\!\!\!\!\!\!\!&&R_2             \leq I(T,U, V ; Y_2, S_2|X_1) - I(T,U, V ; S_1|X_1),\\\label{eqn48}
\!\!\!\!\!\!\!\!\!\!\!\!&&R_1 + R_{2}    \leq  I(T,U, V,X_1; Y_2, S_2) - I(T,U, V ; S_1|X_1),\nonumber\\
\!\!\!\!\!\!\!\!\!\!\!\!&&\\\label{eqn49}
\!\!\!\!\!\!\!\!\!\!\!\!&&R_{e_2}         \leq I(V ; Y_2, S_2|U,X_1, T) \nonumber\\
\!\!\!\!\!\!\!\!\!\!\!\!&&- \max\{I(V ; S_1|U,X_1, T), I(V ; Y_1|U,X_1, T)\},
\end{eqnarray}
for input distribution factors as
\begin{eqnarray}{}\label{eqn50}
&&P_{X_1 S_1 S_2 T U V X_2 Y_1 Y_2} = P_{S_1} P_{ S_2} P_{X_1}\nonumber\\
&& \times P_{T U V X_2|X_1 S_1} P_{Y_1 Y_2|X_1 X_2 S_1 S_2},
\end{eqnarray}
in which the r.h.s. of the equations~\eqref{eqn44}--\eqref{eqn49} are non-negative and~$T, U, V$ are auxiliary random
variables.
\end{theorem}

\begin{proof}
The proof is relegated to Appendix~\ref{D}.
\end{proof}

\subsubsection{The Symmetric Channel State}
For the special case~$S_1 = S_2 = S$, we have the following result.
\begin{col}
For the case in which~$S_1 = S_2 = S$, the set of rates~$(R_1,R_2,R_{e_2} )$ is achievable if it satisfies
\begin{eqnarray}{}\label{eqn58}
\!\!\!\!\!\!\!\!\!&&R_1             \leq I(U, X_1; Y_1) - I(U; S|X_1),\\\label{eqn59}
&&R_{2}           \leq I(X_2 ; Y_2|X_1, S)\\\label{eqn60}
&&R_1 + R_{2}     \leq I(U,X_1; Y_1) + I(X_2; Y_2|X_1, S) \nonumber\\
&&- I(U ; S|X_1),\\\label{eqn61}
&&R_{e_2}         \leq \min\{I(X_2 ; Y_2|U, X_1, S) , \nonumber\\
&&I(X_2 ; Y_2, S|U, X_1) - I(X_2 ; Y_1|U, X_1)\},
\end{eqnarray}
for input distribution that factors as
\begin{eqnarray}{}\label{eqn62}
P_{S X_1 X_2 Y_1 Y_2} = P_{S} P_{X_1}  P_{X_2|X_1 S} P_{Y_1 Y_2|X_1 X_2 S}.
\end{eqnarray}
\end{col}
\begin{proof}
The proof follows directly from Theorem~\ref{thm7}, by substituting~$T = X_1, V = X_2$ and~$S_1 = S_2 = S$.
\end{proof}

\subsection{Outer Bound}
Now, we provide an outer bound on the capacity of the CIC-PCSI, as follows.
\begin{theorem}[Outer bound 3]\label{thm8}
An outer bound on the capacity of the CIC-PCSI consists of the rate pairs~$(R_1,R_2)$ satisfying
\begin{eqnarray}{}\label{eqn63}
\!\!\!\!\!\!\!\!\!\!\!\!\!\!\!\!\!\!\!\! &&R_1             \leq I(T, U, X_1; Y_1) - I(T, U; S_1|X_1),\\\label{eqn64}
\!\!\!\!\!\!\!\!\!\!\!\!\!\!\!\!\!\!\!\! &&R_{2}           \leq I(T, V ; Y_2, S_2|X_1)-I(T, V; S_1|X_1),\\\label{eqn65}
\!\!\!\!\!\!\!\!\!\!\!\!\!\!\!\!\!\!\!\! &&R_1 + R_{2}     \leq I(T, V, X_1; Y_2, S_2) - I(T, V ; S_1|X_1),
\end{eqnarray}
for input distribution that factors as
\begin{eqnarray}{}\label{eqn57}
&&P_{X_1 S_1 S_2 T V X_2 Y_1 Y_2} =  P_{S_1} P_{S_2} P_{X_1} P_{T V X_2|X_1 S_1}\nonumber\\
&&\times P_{Y_1 Y_2|X_1 X_2 S_1 S_2}.
\end{eqnarray}
\end{theorem}
\begin{proof}
The proof is similar to one taken by~\cite[Appendix F]{bibi82} by using the Fano's inequality.
\end{proof}

\subsection{The Gaussian Example}
In this section, we consider
the CIC with partial channel states (shown in Figure~\ref{fig:2}) with
the channel outputs as~\eqref{eqn37}. Similar to the cases considered in Section III-C, when~$a > 1$, i.~e., \textit{Strong Interference},
we have no secrecy. Thus, we consider the other case~$a \leq 1$. We provide the following theorem for the Gaussian
CIC-PCSI.
\begin{theorem}[Achievable equivocation-rate region]\label{thm9}
For the Gaussian CIC-PCSI, in the case that~$a \leq 1$, the achievable equivocation-rate region consists of the rate triples~$(R_1, R_2, R_{e_2})$
which satisfy~\eqref{eqn66}--\eqref{eqn70}, in the above of the page,
\begin{figure*}
\begin{eqnarray}{}\label{eqn66}
R_1      &\leq& \mathcal{C}\Big(\frac{P_1+a^2P_2+ K_1+K_2+1+ 2a\rho_1\sqrt{P_1 P_2}+ 2a\rho_2\sqrt{P_2 K_1}}{K_1(a^2 P_2^{''}+ K_2+1)}\Big),\\\label{eqn67}
R_2      &\leq& \mathcal{C}\big(P_2^{''}\big),\\\label{eqn69}
R_1+R_2  &\leq& \mathcal{C}\big(b^2P_1+ P_2+ K_1+ 2b\rho_1\sqrt{P_1 P_2}+ 2\rho_2\sqrt{P_2 K_1}\big) -\frac{1}{2}\log\big(K_1\big) ,\\\label{eqn70}
R_{e_2} &\leq&  \mathcal{C}\big(P_2^{''}\big)- \mathcal{C}\big(\frac{a^2P_2^{''}}{K_2+1}\big),
\end{eqnarray}
\end{figure*}
in which~$P_2^{''}=\rho P_2$ and~$0\leq \rho, \rho_1,\rho_2 \leq 1$.
\end{theorem}
\begin{proof}
The proof is based on Theorem~\ref{thm7}, by substituting~$T = (U, S_1)$ and~$V = X_2$ and choosing the following jointly Gaussian distributions for the random variables:
\begin{eqnarray}{}\label{eqn76}
&&\!\!\!\!\!\! X_1\sim\mathcal{N}(0, P_1),~X_2\sim\mathcal{N}(0, P_2),\\\label{eqn77}
&&\!\!\!\!\!\! X_2= X_2^{'}+X_2^{''}+\rho_1\sqrt{\frac{P_2}{P_1}}X_1+ \rho_2\sqrt{\frac{P_2}{K_1}}S_1,\\\label{eqn78}
&&\!\!\!\!\!\! X_2^{'}\sim\mathcal{N}(0, P_2^{'}),~X_2^{''}\sim\mathcal{N}(0, P_2^{''}),\\\label{eqn79}
&&\!\!\!\!\!\! P_2^{'}+ P_2^{''} =(1-\rho_1^2-\rho_2^2)P_2,\\\label{eqn80}
&&\!\!\!\!\!\! U=X_2^{'}+\alpha S_1,
\end{eqnarray}
in which~$X_1, X_2^{'}, X_2^{''}$, and~$S_1, S_2$ are independent. Transmitter~2, splits its power into three parts:~$\rho_1^2P_2$, which
is used for cooperating with the primary transmitter sending~$W_1$;~$P_2^{'}+\rho_2^2P_2$,
which is used in dirty paper coding to deal with
the state at receiver~1 via an auxiliary random variable~$U$; and~$P^{''}_2$
which is used for transmitting~$W_2$.
The mutual information formulas in~\eqref{eqn66}-\eqref{eqn70} are calculated by the approach taken by~\cite{bibi82}.
\end{proof}

To compare the results of Theorem~\ref{thm9} with the achievable equivocation-rate of Theorem~\ref{thm5}, we consider
a simple case of Theorem~\ref{thm9} in which the cognitive transmitter uses all its power to send
its individual message. For this case we have the following corollary.
\begin{col}[Perfect Secrecy Condition]\label{col3}
For the Gaussian CIC-PCSI, in the case that~$a \leq 1$, the achievable secrecy rate region
consists the set of rates~$(R_1, R_2)$ which satisfies
\begin{eqnarray}{}\label{eqn661}
R_1      &\leq& \mathcal{C}\big(\frac{P_1+a^2P_2+ K_1+K_2+1}{K_1(a^2 P_2 + K_2+1)}\big),\\\label{eqn671}
R_{2}    &\leq& \mathcal{C}\big(P_2\big)- \mathcal{C}\big(\frac{a^2P_2}{K_2+1}\big),\\\label{eqn701}
\!\!\!\!\!\! R_1+R_2  &\leq& \mathcal{C}\big(b^2P_1+ P_2+ K_1\big) -\frac{1}{2}\log K_1.
\end{eqnarray}
\end{col}
\begin{proof}
The proof is directly derived from Theorem~\ref{thm9} by considering the perfect secrecy condition in which~$R_{2}\leq \min\{R_{2}, R_{e_2}\}$, and
by substituting~$\rho_1=\rho_2=0$ and~$\rho=1, X_2^{'}=\emptyset$,
\end{proof}

\begin{remark}
Comparing the results of Corollary~\ref{col3} with the achievable equivocation-rate region of Theorem~\ref{thm5} shows that
the secrecy rate of~\eqref{eqn671} is higher than the one of~\eqref{eqn40}, because of~$K_2\geq 0$. It means that
the SPC approach achieves higher secrecy rate than the GPC in general case. Moreover, comparing~\eqref{eqn661} with~\eqref{eqn38} shows that for the case of~$a\leq a^{\dag}$, where
\begin{eqnarray}{}\label{eqn662}
a^{\dag}= \sqrt{\frac{(K_2+1)(P_1+K_1+K_2+1)- P_1K_1(K_2+1)}{P_1P_2K_1-P_2(K_2+1)}}\nonumber
\end{eqnarray}
the SPC approach obtains higher achievable rate for the primary transmitter than the GPC, and for the case of~$a> a^{\dag}$ vice versa.
Thus, in the case of~$a> a^{\dag}$, there is a trade off between the secrecy rate of the cognitive transmitter and the achievable rate of the
primary one. Figure~\ref{fig:3} shows the secrecy rate of the cognitive transmitter vs. the achievable rate of the primary one, using the GPC and the SPC approaches, and it illustrates the trade off between the~$R_2$ and~$R_1$ in the case of~$a> a^{\dag}$.
\end{remark}

\begin{figure*}
\centering
\epsfig{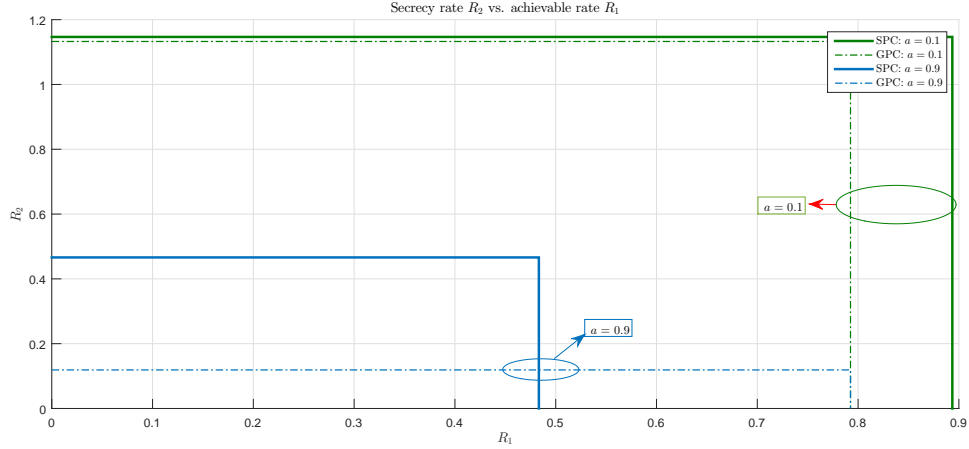}
\caption{The achievable equivocation-rate region of Theorem~\ref{thm5} (GPC),
and the achievable secrecy rate region of Theorem~\ref{thm9} (SPC)
for special values~$P_1= 4, P_2= 4, K_1= K_2= 1, b=0.3$ and $a=0.1, 0.9$. In this case~$a^{\dag}= 0.866$, and for~$a> a^{\dag}$,
there is a trade off between the secrecy rate of the cognitive transmitter and the achievable rate of the
primary one.}
\label{fig:3}
\end{figure*}

\section{Conclusions}\label{S5}

In this paper we studied the Cognitive Interference Channel
in which the partial channel state's information is available non-causally at the
cognitive transmitter and corresponding receiver.
Furthermore, the cognitive transmitter wishes to keep its message
confidential at the primary receiver, in addition to have a reliable
communication with its destination.
We use the Gel'fand-Pinsker coding (GPC) and the superposition
coding (SPC) to show that how the cognitive transmitter can
use the side information about the primary message and the channel state sequence to improve its achievable
rate and cooperate with the primary one. Therefore, we have derived the achievable equivocation-rate region for this channel
in two cases: by using GPC and SPC.
Moreover, in each case the outer bounds on the capacity and extension to a simple Gaussian example is presented.
In the Gaussian case, we consider a case in which the partial channel state sequences are additive and independent Gaussian
random variables, and it is shown that in some cases, there is a trade off between the secrecy rate of the cognitive transmitter
and the achievable rate of the primary one, using the GPC and GPC approaches. Thus, the cognitive transmitter can obtain the desired
region by choosing the proper coding scheme.

\appendix
\section{Proof of Theorem~\ref{thm1}}\label{A}
\begin{proof}
The proof is established on the proof of~\cite[Theorem 1]{bibi13} and~\cite[Theorem 1]{bibi31}.
After introducing the code-book generation, and the encoding-decoding scheme, the proof of
Theorem~\ref{thm1} is presented in two steps. In step I, we prove the reliability of the rate region, i.~e.,
the condition under which the probability of error tends to zero for~$n\rightarrow\infty$. This step yields to the equations~\eqref{eqn10}-\eqref{eqn15}. In step II, we will calculate the equivocation to evaluate
the secrecy level. This step provides the equation~\eqref{eqn16}.

\textbf{Code-book generation:}
\begin{enumerate}

\item
For split rates~\eqref{eqn8}-\eqref{eqn9}, generate~$2^{nR_{1a}}$ codewords~$x_{1a}^n (w_{1a})$,~$w_{1a}\in \{1, 2, \ldots,2^{nR_{1a}}\}$, choosing~$x_{1a,n}^n (w_{1a})$ independently according to~$P_{X_{1a}}(\cdot)$.
\item For each~$w_{1a}$, generate~$2^{nR_{1b}}$ codewords~$x_{1b}^{n}(w_{1a}, w_{1b})$ using~$\Pi_{i=1}^{n} P_{X_{1b}\mid X_{1a}}(\cdot\mid x_{1a,i}(w_{1a}))$, where~$w_{1b}\in\{ 1,2, \ldots, 2^{nR_{1b}}\}$.
\item Over each pair $(w_{1a}, w_{1b})$ we generate~$x_{1}^{n}(w_{1a}, w_{1b})$ where~$x_{1}$ is a deterministic function of~$(x_{1a},x_{1b})$.
\item Generate $2^{n(R_{2b}+ R_{2b}^{'})}$ codewords $u^{n}(w_{2b}, b_{2b})$, $w_{2b}\in\{1,2,\ldots, 2^{nR_{2b}}\}$,~$b_{2b}\in\{1,2, \ldots, 2^{nR_{2b}^{'}}\}$ using~$P_{U}(\cdot)$.
\item For each $u^{n}(w_{2b}, b_{2b})$ generate $2^{n(R_{2a}+ R_{2a}^{'})}$ codewords $v^{n} (w_{2b}, b_{2b},w_{2a} ,b_{2a})$,
$w_{2a}\in\{1, 2, \ldots, 2^{nR_{2a}}\}, b_{2a}\in\{1, 2, \ldots, 2^{nR_{2a}^{'}}\}$ using
$\Pi_{n=1}^{n} P_{V\mid U}(\cdot\mid u_{n}(w_{2b},b_{2b}))$.
\item Now, define~$L_{1} = I(V;Y_{2}, S_{2}|U)-I(V;Y_{1}, X_{1}|U),~L_{2} = I(V;Y_{1},X_{1}|U)$.
Note that, here we assume that $R_{2a}>L_1\geq 0$, for the case $R_{2a}<L_1$
the similar coding scheme can be used to obtain the
\textit{perfect secrecy}, which is mentioned at the end of the proof. Let
\begin{eqnarray}{}\label{eqn82}
\mathcal{W}_{2a}=\mathcal{A}\times\mathcal{B}
\end{eqnarray}
where~$\mathcal{A}=\{1,2,\ldots,2^{n(R_{2a}-L_{1})}\}$ and~$\mathcal{B}=\{1,2,\ldots,2^{nL_{1}}\}$. Then, we define the mapping~$f:\mathcal{B}\rightarrow\mathcal{C}$ to partition~$\mathcal{B}$ into~$2^{nL_{1}}$ subsets with nearly equal size which means
\begin{equation}\label{eqn84}
\|f^{-1}(c_{1})\| \leq 2\|f^{-1}(c_{2})\| ~ \textrm{for~ each}~ c_{1},c_{2}\in \mathcal{C}.
\end{equation}

Now we define the mapping~$w_{2a}=(a,c)\rightarrow (a,b)$, in which~$b$ is chosen randomly from the set~$f^{-1}(c)\subset \mathcal{B}$.

\item Over each pair~$w_{1}$ and~$w_{2}(w_{2a}(a,c),w_{2b})$, we generate~$x_{1}^{n}$ and~$x_{2}^{n}(w_{1}, w_{2b}, w_{2a}(a,c), b_{2b}, b_{2a}, s_{1})$ where~$x_{2}$ is a deterministic function of~$(u, v, x_{1}, s_1)$.
\end{enumerate}

\textbf{Encoding:}
\begin{enumerate}

\item
Encoder~2 splits the~$nR_{2}$ bits~$w_{2}$ into~$nR_{2a}$ bits~$w_{2a}$ and~$nR_{2b}$ bits~$w_{2b}$. Similarly, it splits the~$nR_{1}$ bits~$w_{1}$ into~$nR_{1a}$ bits~$w_{1a}$ and~$nR_{1b}$ bits~$w_{1b}$. Thus,
\begin{eqnarray}{}\label{eqn85}
w_{2}=(w_{2a},w_{2b}),~~w_{1}=(w_{1a},w_{1b}).
\end{eqnarray}

\item  Encoder~2, finds a bin index~$b_{2b}$ such that~$(u^{n} (w_{2b}, b_{2b}),x_{1a}^{n}(w_{1a}),x_{1b}^{n} (w_{1a}, w_{1b}), s_{1}^{n})$
are jointly typical. If such a~$b_{2b}$ is not found, it chooses $b_{2b}=1$.

\item  For each~$(w_{2b}, b_{2b})$ and given~$s_{1}^{n}$ encoder~2 finds a bin index~$b_{2a}$ such that
$ (v^{n} (w_{2b},b_{2b}, w_{2a},b_{2a}), $ $u^{n} (w_{2b}, b_{2b}), x_{1a}^{n}(w_{1a}), x_{1b}^{n} (w_{1a}, w_{1b}), $ $ s_{1}^{n})$ are jointly typical.

\item  Transmitter~1 transmits~$x_{1}^{n} (w_{1a}, w_{1b})$.

\item  Transmitter~2 transmits $x_{2}^{n} (w_{1a}, w_{1b}, w_{2a}, b_{2a}, w_{2b}, b_{2b}, s_{1}^{n})$.
\end{enumerate}

\textbf{Decoding:}
\begin{enumerate}
\item  For given~$y_{1}^{n}$, decoder~1 chooses~$(\hat{w}_{1a}, \hat{w}_{1b}, \hat{w}_{2b}, \hat{b}_{2b})$ such that
$\Big(u^{n} (\hat{w}_{2b}, \hat{b}_{2b}), x_{1a}^{n} (\hat{w}_{1a}), x_{1b}^{n} (\hat{w}_{1a}, \hat{w}_{1b}), y_{1}^{n}\Big)$ are jointly typical. If there is no such quadruple it chooses~$(1,1,1,1)$.
\item
For given~$y_{2}^{n}$,~$s_{2}^{n}$ decoder~2 chooses $(\hat{w}_{2b}, \hat{b}_{2b}, \hat{w}_{2a}, \hat{b}_{2a})$ such that
$\Big(v^{n} (\hat{w}_{2b}, \hat{b}_{2b}, \hat{w}_{2a}, \hat{b}_{2a}), u^{n} (\hat{w}_{2b}, \hat{b}_{2b}), y_{2}^{n}, s_{2}^{n}\Big)$
are jointly typical. If there are more than one such quadruple, it chooses one of them. If there is not any quadruple, it chooses~$(1,1,1,1)$.
\end{enumerate}

\subsection{Step I: (Reliability) achievability of the rate region~\eqref{eqn10}--\eqref{eqn15}}
Reliability of the rate region~\eqref{eqn10}--\eqref{eqn15} will be proved here by analyzing the error probability.

\textbf{Error analysis:}
Using this scheme for coding and decoding, analysis of the error is derived following~\cite{bibi13}. First, we suppose that~$(w_{2a},w_{2b},w_{1a},w_{1b}) =(1, 1, 1, 1)$ is sent. An encoder error occurs in one of the following situations.

1-~$\mathcal{E}_{1}$: Encoder~2, cannot find a bin index~$b_{2b}$ such that~$(u^{n} (1, b_{2b}),x_{1a}^{n}(1),x_{1b}^{n} (1,1), s_{1}^{n})\in T_{\epsilon}^{(n)}(P_{U,X_{1a},X_{1b}, S_{1}})$ in which~$T_{\epsilon}^{(n)}(P_{XY})$ denotes the jointly~$\epsilon$-typical set with respect to~$P_{XY}$. It can be shown, by covering lemma~\cite{bibi69}, that for~$n\rightarrow\infty$ such~$b_{2b}$ exists with high probability if we have
\begin{eqnarray}{}\label{eqn86}
R^{'}_{2b}>I(U;X_{1a},X_{1b}, S_{1})+\delta,
\end{eqnarray}
in which~$\delta$ tends to zero as~$n\rightarrow \infty$~\cite{bibi13}.

2-~$\mathcal{E}_{2}$: After finding $b_{2b}=1$, encoder~2 cannot find~$b_{2a}$ such that
$(v^{n} (1,1,1,b_{2a}), u^{n} (1, 1), x_{1a}^{n} (1), x_{1b}^{n} (1, 1), s_{1}^{n})\in T_{\epsilon}^{(n)}(P_{U, V,X_{1a},X_{1b}, S_{1}})$.
It can be shown~\cite{bibi13} that for~$n\rightarrow\infty$, such~$b_{2a}$ exists with high probability if we have
\begin{eqnarray}{}\label{eqn87}
R^{'}_{2a}> I(V;X_{1a},X_{1b}, S_{1}|U)+\delta.
\end{eqnarray}

Now, we should compute the probabilities of the error events at the decoder, which are shown in TABLE~\ref{tab:1}. The second column of the table shows the corresponding bounds of the rates, which can be shown, make the error probability of each event tend to zero, as~$n\rightarrow\infty$. For more details about the derivation of the bounds proposed in TABLE~\ref{tab:1} see~\cite{bibi13}.
Using these bounds, the achievability of~\eqref{eqn10}--\eqref{eqn15} are proved.

\begin{table*}[tbp]
\centering
\caption{Error events in joint decoding and corresponding rate bounds}
\begin{tabular}[tbp]{| l | l | l |}\hline
& Error event & Arbitrarily small positive error probability if \\\hline
$E_{1}$ & $(\hat{w}_{2b}\neq 1, \hat{w}_{2a}= 1)$                             & $R_{2b}+ R_{2b}^{'} \leq I(U, V; Y_{2}, S_{2})$\\\hline
$E_{2}$ & $(\hat{w}_{2b}= 1, \hat{w}_{2a}\neq 1)$                             & $R_{2a}+ R_{2a}^{'} \leq I(V; Y_{2}, S_{2}| U)$\\\hline
$E_{3}$ & $(\hat{w}_{2b}\neq 1, \hat{w}_{2a}\neq 1)$                          & $R_{2b} + R_{2b}^{'} + R_{2a}+ R_{2a}^{'} \leq I(U, V; Y_{2}, S_{2})$\\\hline
$E_{1}^{'}$ & $(\hat{w}_{1a}\neq 1, \hat{w}_{1b}= 1, \hat{w}_{2b}^{'}= 1)$    & $R_{1a} \leq I(X_{1a}, X_{1b}; U, Y_{1})$\\\hline
$E_{2}^{'}$ & $(\hat{w}_{1a}\neq 1, \hat{w}_{1b}\neq 1, \hat{w}_{2b}^{'}= 1)$ & $R_{1a}+ R_{1b} \leq I(X_{1a}, X_{1b}; U, Y_{1})$\\\hline
$E_{3}^{'}$ & $(\hat{w}_{1a}\neq 1, \hat{w}_{1b}= 1, \hat{w}_{2b}^{'}\neq 1)$ & $R_{1a}+R_{1b} + R_{1b}^{'} \leq I(U, X_{1a}, X_{1b}; Y_{1})+ I(U; X_{1a}, X_{1b})$\\\hline
$E_{4}^{'}$ & $(\hat{w}_{1a}\neq 1, \hat{w}_{1b}\neq 1, \hat{w}_{2b}^{'}\neq 1)$ & $R_{1a}+R_{1b}+R_{2b} + R_{2b}^{'} \leq I(U, X_{1a}, X_{1b}; Y_{1})+ I(U; X_{1a}, X_{1b})$\\\hline
$E_{5}^{'}$ & $(\hat{w}_{1a}= 1, \hat{w}_{1b}\neq 1, \hat{w}_{2b}^{'}= 1)$    & $R_{1b} \leq I(X_{1b}; Y_{1}, U| X_{1a})$\\\hline
$E_{6}^{'}$ & $(\hat{w}_{1a}= 1, \hat{w}_{1b}\neq 1, \hat{w}_{2b}^{'}\neq 1)$ & $R_{1b}+ R_{2b} + R_{2b}^{'} \leq I(X_{1b}, U; Y_{1}| X_{1a})+ I(U; X_{1a}, X_{1b})$\\\hline
\end{tabular}
\label{tab:1}
\end{table*}

\subsection{Step II: (Security) achievability of the equivocation-rate region~\eqref{eqn16}}

In this step, the achievability of the equivocation-rate region~\eqref{eqn16} will be driven. To this purpose, we compute the equivocation.

\textbf{Equivocation-rate calculation:}
To prove~\eqref{eqn16}, for the equivocation-rate $R_{e_2}$, we follow the proof reported in~\cite{bibi31,bibi8,bibi9}.
We establish computing of the equivocation for the cognitive transmitter as follows.
\begin{eqnarray}{}\label{eqn88}
&&H(W_{2a}, W_{2b}\mid Y_{1}^{n})\nonumber\\
&\geq& H(W_{2a}, W_{2b}\mid Y_{1}^{n}, W_{1}, W_{2b})\nonumber\\
&=& H(W_{2a},Y_{1}^{n}\mid W_{1}, W_{2b})- H(Y_{1}^{n}\mid W_{1}, W_{2b})\nonumber\\
&=& H(W_{2a}, Y_{1}^{n}, V^{n}\mid W_{1}, W_{2b})\nonumber\\
&&  -H(V^{n}\mid W_{2a},W_{1}, W_{2b}, Y_{1}^{n})- H(Y_{1}^{n}\mid W_{1}, W_{2b})\nonumber\\
&=& H(W_{2a}, V^{n}\mid W_{1}, W_{2b})\nonumber\\
&&  + H(Y_{1}^{n}\mid W_{1},W_{2b},W_{2a},V^{n})\nonumber\\
&&  -H(V^{n}\mid W_{2a},W_{1}, W_{2b}, Y_{1}^{n})-H(Y_{1}^{n}\mid W_{1}, W_{1b})\nonumber\\
&\stackrel{(a)}{\geq}& H(V^{n}\mid W_{1},W_{2b})+ H(Y_{1}^{n}\mid V^{n},U^{n}, X_{1}^{n})\nonumber\\
&&  -H( V^{n}\mid W_{2a},W_{1},W_{2b},Y_{1}^{n})- H(Y_{1}^{n} \mid W_{1},W_{2b})\nonumber\\
&&
\end{eqnarray}
where~$(a)$ is because of the fact that given~$V^{n}$,~$W_{2a}$ is uniquely determined and~$Y_{1}^{n}$ is independent of~$(W_{1},W_{2b},W_{2a})$ given~$(V^{n},U^{n}, X_{1}^{n})$. Now, we bound each term in r.h.s. of~\eqref{eqn88}. For the first term in~\eqref{eqn88}, we have
\begin{eqnarray}{}\label{eqn89}
&&H(V^{n}\mid W_{1},W_{2b})\nonumber\\
&\stackrel{(b)}{\geq}& H(V^{n}\mid U^{n},X_{1}^{n})\nonumber\\
&\geq& H(V^{n}\mid U^{n},X_{1}^{n})-H(V^{n}\mid U^{n},Y_{2}^{n},S_{2}^{n})\nonumber\\
&=&I(V^{n};Y_{2}^{n},S_{2}^{n} \mid U^{n})-I(V^{n};X_{1}^{n}\mid U^{n})\nonumber\\
&\stackrel{(c)}{\geq}& n[I(V;Y_{2}, S_{2}\mid U)-I(V;X_{1}\mid U)]
\end{eqnarray}
where~$(b)$ is derived by using the data processing
inequality~\cite{bibi34}, which implies that~$V^{n}$ is independent of~$(W_{1},W_{2b})$
given~$(U^{n},X_{1}^{n})$, and~$(c)$ is derived using the approach taken in~\cite[Lemma 3]{bibi15}.
For the second term in the r.h.s of~\eqref{eqn88} we follow the related equations in~\cite{bibi9} and obtain
\begin{equation}\label{eqn90}
\frac{1}{n} H(Y_{1}^{n}\mid V^{n},U^{n}, X_{1}^{n})\geq H(Y_{1}\mid V,U,X_{1}, S_{1})-\epsilon_{1},
\end{equation}
where~$\epsilon_{1}$ is negligible for~$n\rightarrow\infty$. To compute the third term in the r.h.s of~\eqref{eqn88}, similar to~\cite[Lemma 2]{bibi9}, by using Fano's inequality we obtain
\begin{eqnarray}{}\label{eqn91}
\frac{1}{n}H(V^{n}\mid W_{2a},W_{1},W_{2b},Y_{1}^{n})< \epsilon_{2}
\end{eqnarray}
where~$\epsilon_{2}$ is negligible, when~$n\rightarrow\infty$. To compute the fourth term in~\eqref{eqn88}, first we define
\begin{eqnarray}{}\label{eqn92}
\hat{y}_{1}^{n}=
\left\{\begin{array}{ll}
y_{1}^{n}&\textrm{if} ~~\Big(u^{n}({w}_{2b}, {b}_{2b}), x_{1a}^{n}(w_{1a}), \\
& x_{1b}^{n}(w_{1a}, w_{1b}), y_{1}^{n}\Big)\in T_{\epsilon}^{(n)}(P_{UX_{1}Y_{1}})\\
z^{n}&\textrm{Otherwise} \\
\end{array} \right.
\end{eqnarray}
where~$z^{n}$ is an arbitrary sequence that is contained in $\mathcal{Y}_{1}^n$. Now, we have
\begin{eqnarray}{}\label{eqn93}
\!\!\!\!\!\!\!\!\!\!&&\frac{1}{n}H(Y_{1}^{n}\mid W_{1},W_{2b})\nonumber\\
\!\!\!\!\!\!\!\!\!\!&=&   \frac{1}{n}\sum_{w_{1},w_{2b}}[Pr\{W_{1}=w_{1},W_{2b}=w_{2b}\} \nonumber\\
\!\!\!\!\!\!\!\!\!\!&&H(Y_{1}^{n}\mid W_{1}=w_{1},W_{2b}=w_{2b})]\nonumber\\
\!\!\!\!\!\!\!\!\!\!&\leq&\frac{1}{n}\sum_{w_{1},w_{2b}}[Pr\{W_{1}=w_{1},W_{2b}=w_{2b}\} \nonumber\\
\!\!\!\!\!\!\!\!\!\!&&H(\hat{Y}_{1}^{n},Y_{1}^{n}\mid W_{1}=w_{1},W_{2b}=w_{2b})]\nonumber\\
\!\!\!\!\!\!\!\!\!\!&=&   \frac{1}{n}\sum_{w_{1},w_{2b}} Pr\{W_{1}=w_{1},W_{2b}=w_{2b}\} \nonumber\\
\!\!\!\!\!\!\!\!\!\!&&\Big[H(\hat{Y}_{1}^{n}\mid W_{1}=w_{1},W_{2b}=w_{2b})\nonumber\\
\!\!\!\!\!\!\!\!\!\!&&+H(Y_{1}^{n}\mid W_{1}=w_{1},W_{2b}=w_{2b},\hat{Y}_{1}^{n})\Big]
\end{eqnarray}

For the first term in ~\eqref{eqn93} we can write
\begin{eqnarray}{}\label{eqn94}
\!\!\!\!\!\!\!\!\!\!&&\frac{1}{n}    \sum_{w_{1},w_{2b}}Pr\{W_{1}=w_{1},W_{2b}=w_{2b}\}\nonumber\\
\!\!\!\!\!\!\!\!\!\!&&H(\hat{Y}_{1}^{n}\mid W_{1}=w_{1},W_{2b}=w_{2b})\nonumber\\
\!\!\!\!\!\!\!\!\!\!&\stackrel{(d)}\leq&\frac{1}{n}\sum_{w_{1},w_{2b}}Pr\{W_{1}=w_{1},W_{2b}=w_{2b}\}\nonumber\\
\!\!\!\!\!\!\!\!\!\!&&\times \log| T_{\epsilon}^{(n)}(P_{Y_{1}\mid U,X_{1}})|\nonumber
\end{eqnarray}
\begin{eqnarray}
&\leq&           \sum_{w_{1},w_{2b}}Pr\{W_{1}=w_{1},W_{2b}=w_{2b}\} \nonumber\\
&&\times [H(Y_{1}\mid U,X_{1})+\epsilon_{3}]\nonumber\\
&\leq& H(Y_{1}\mid U,X_{1})+\epsilon_{3},
\end{eqnarray}
where~$(d)$ is based on AEP~\cite{bibi34}, and~$\epsilon_{3}$ is negligible for~$n\rightarrow\infty$. To bound the second term in the r.h.s of~\eqref{eqn93}, we use Fano's inequality and obtain
\begin{eqnarray}{}\label{eqn95}
&&\frac{1}{n}      \sum_{w_{1},w_{2b}} Pr\{W_{1}=w_{1},W_{2b}=w_{2b}\}\nonumber\\
&&H(Y_{1}^{n}\mid W_{1}=w_{1},W_{2b}=w_{2b},\hat{Y}_{1}^{n})\nonumber\\
&\leq&\frac{1}{n}\sum_{w_{1},w_{2b}}   Pr\{W_{1}=w_{1},W_{2b}=w_{2b}\}\nonumber\\
&&\Big(1+Pr\{Y_{1}^{n}\neq\hat{Y}_{1}^{n}\mid W_{1}=w_{1},W_{2b}=w_{2b}\}\nonumber\\
&&\times \log{|\mathcal{Y}_{1}|^{n}}\Big)\nonumber\\
&=&\frac{1}{n}+  \log|\mathcal{Y}_{1}|\sum_{w_{1},w_{2b}} Pr\{W_{1}=w_{1},W_{2b}=w_{2b}\}\nonumber\\
&& Pr\{Y_{1}^{n}\neq \hat{Y}_{1}^{n}\mid W_{1}=w_{1},W_{2b}=w_{2b}\}\nonumber\\
&\leq&\epsilon_{4},
\end{eqnarray}
where~$\epsilon_{4}$ is negligible for~$n\rightarrow\infty$. Hence, from~\eqref{eqn94} and~\eqref{eqn95}, the forth term of the r.h.s. of~\eqref{eqn88} is bounded as
\begin{eqnarray}{}\label{eqn96}
\frac{1}{n} H(Y_{1}^{n}\mid W_{1},W_{2b})
\leq H(Y_{1}\mid U,X_{1})+\epsilon_{5},
\end{eqnarray}
in which $\epsilon_{5}$ tends to zero for~$n\rightarrow\infty$. Substituting~\eqref{eqn89},~\eqref{eqn90},~\eqref{eqn91} and~\eqref{eqn96} into~\eqref{eqn88}, we obtain
\begin{eqnarray}{}\label{eqn97}
\frac{1}{n}&H&(W_{2a},W_{2b}\mid Y_{1}^{n})\nonumber\\
&\geq& I(V;Y_{2}, S_{2}, U)-I(V;X_{1}, U) \nonumber\\
&&    +H(Y_{1}\mid V,U,X_{1}, S_{1})- H(Y_{1}\mid U, X_{1})-\epsilon_{6}\nonumber\\
&\geq& I(V; Y_{2}, S_{2}, U)-I(V, S_{1};X_{1},U) \nonumber\\
&&    -I(Y_1; V, S_1| U, X_1)-\epsilon_{6}\nonumber\\
&=&    I(V;Y_{2},S_{2}, U)-I(V, S_{1}; X_{1},Y_{1}, U)-\epsilon_{6}\nonumber\\
&&
\end{eqnarray}
where~$\epsilon_{6}$ is negligible for~$n\rightarrow\infty$. Regard to the definition of~$R_{e_2}$ in~\eqref{eqn5}-\eqref{eqn6} we conclude
\begin{eqnarray}{}\label{eqn98}
R_{e_2} \leq I(V;Y_{2},S_{2}, U)-I(V, S_{1};X_{1},Y_{1},U).
\end{eqnarray}
and therefore~\eqref{eqn16} is proved.
\end{proof}

\section{Proof of Theorem~\ref{thm3}}\label{B}
\begin{proof}[Proof of Theorem~\ref{thm3}]
For a quadruple code~$(M_1, M_2, n, P_e)$ for the CIC-PCSI, we consider the outer bound on~$R_1$ proposed
in~\eqref{eqn24}. Using the Fano's inequality we have
\begin{eqnarray}{}\label{eqn99}
nR_{1} &\leq& I(W_{1}; Y_{1}^{n})\nonumber\\
&=&    \sum_{i=1}^{n} I(W_{1}; Y_{1, i} | Y_{1, i+1}^{n})\nonumber\\
&\leq& \sum_{i=1}^{n} I(W_{1}, Y_2^{i-1} , Y_{1,i+1}^{n}; Y_{1,i})\nonumber\\
&\stackrel{(e)}{=}& \sum_{i=1}^{n} I(W_{1}, U_{i}; Y_{1,i}),
\end{eqnarray}
where~$(e)$ is derived by substituting~$U_{i}=(Y_{2}^{i-1}, Y_{1, i+1}^{n})$. Then, by substituting~$V_{1,i}= W_{1}$, the outer bound on~$R_1$ is derived. Similarly, we have
\begin{eqnarray}{}\label{eqn100}
nR_{1} &\leq& I(W_{1}; Y_{1}^{n})\nonumber\\
&\leq& \sum_{i=1}^{n} I(W_{1}; Y_2^{i-1} , Y_{1,i+1}^{n}, Y_{1,i})\nonumber\\
&\stackrel{(f)}{=}& \sum_{i=1}^{n} I(W_{1}; Y_{1,i}, U_{i}),
\end{eqnarray}
where~$(f)$ is derived by substituting~$U_{i}=(Y_{2}^{i-1}, Y_{1, i+1}^{n})$.
Thus, the outer bound on~$R_1$ is proved. The outer bound for~$R_2$ is derived as follows:

\begin{eqnarray}{}\label{eqn101}
  nR_{2} &\leq& I(W_{2}; Y_{2}^{n}|S_{1}^{n},S_{2}^{n})\nonumber\\
&=& \sum_{i=1}^{n} I(W_{2}, Y_{1, i+1}^n; Y_{2, i}| Y_{2}^{i-1}, S_{1}^{n},S_{2}^{n})\nonumber\\
&&- I(Y_{1, i+1}^n; Y_{2}^{i}| Y_{2}^{i-1}, S_{1}^{n},S_{2}^{n}, W_{2})\nonumber\\
&=& \sum_{i=1}^{n} I(W_{2}; Y_{2, i}| Y_{2}^{i-1},  Y_{1, i+1}^n, S_{1}^{n},S_{2}^{n})\nonumber\\
&&+ I(Y_{1, i+1}^n; Y_{2, i}| Y_{2}^{i-1}, S_{1}^{n},S_{2}^{n})\nonumber\\
&&- I(Y_{1, i+1}^n; Y_{2, i}| Y_{2}^{i-1}, S_{1}^{n},S_{2}^{n}, W_{2})\nonumber\\
&\stackrel{(g)}{=}& \sum_{i=1}^{n} I(W_{2}; Y_{2, i}| Y_{2}^{i-1},  Y_{1, i+1}^n, S_{1}^{n},S_{2}^{n})\nonumber\\
&&+ I(Y_{2}^{i-1}; Y_{1, i}, W_2| Y_{1, i+1}^{n}, S_{1}^{n},S_{2}^{n})\nonumber\\
&&- I(W_2; Y_{2}^{i-1}| Y_{1, i+1}^{n}, Y_{1, i}, S_{1}^{n},S_{2}^{n})\nonumber\\
&& - I(Y_{2}^{i-1}; Y_{1, i}| Y_{1, i+1}^{n}, S_{1}^{n},S_{2}^{n}, W_{2})\nonumber
\end{eqnarray}
\begin{eqnarray}
&=& \sum_{i=1}^{n} I(W_{2}; Y_{2}^{i}| Y_{1, i+1}^{n}, S_{1}^{n},S_{2}^{n})\nonumber\\
&& -I(W_{2}; Y_{2}^{i-1}| Y_{1, i}^{n}, S_{1}^{n},S_{2}^{n})\nonumber\\
&\leq& \sum_{i=1}^{n} I(W_{2}; Y_{2}^{i}| Y_{1, i+1}^{n}, S_{1}^{n},S_{2}^{n})\nonumber\\
&=& \sum_{i=1}^{n} I(W_{2}, U_{i}; Y_{2,i}| S_{1}^{n},S_{2}^{n}),
\end{eqnarray}
where~$(g)$ is derived by Csisz\'{a}r sum identity~\cite{bibi69}. Then, by substituting~$V_{1,i}= W_{1}$ and~$V_{2,i}= W_{2}$, the outer bound on~$R_2$ is derived. From Fano's inequality~\cite[Chapter 7]{bibi34} we have
\begin{eqnarray}{}\label{eqn102}
&&n(R_1+R_2)\nonumber\\
&\leq&                 I(W_{1}; Y_{1}^{n})+ I(W_{2}; Y_{2}^{n}|S_{1}^{n}, S_{2}^{n})\nonumber\\
&\stackrel{(h)}{\leq}& I(W_{1}; Y_{1}^{n})+I(W_{2}; Y_{2}^{n}| W_{1},S_{1}^{n}, S_{2}^{n})\nonumber\\
&=&\sum_{i=1}^{n}      I(W_{1}; Y_{1, i}| Y_{1, i+1}^{n}) \nonumber\\
&&+ I(W_2; Y_{2}^{i}| W_{1}, Y_{1, i+1}^{n},S_{1}^{n}, S_{2}^{n})\nonumber\\
&&-\Big[I(W_{2}, Y_{1,i} ;Y_{2}^{i-1}|W_1, Y_{1,i}^{n})\nonumber\\
&&- I(Y_{1,i} ;Y_{2}^{i-1}|W_1, Y_{1,i}^{n})\Big]\nonumber\\
&\stackrel{(i)}{\leq}&\sum_{i=1}^{n} I(W_{1}, Y_{2}^{i-1}; Y_{1,i}| Y_{1, i+1}^{n}) \nonumber\\
&&+ I(W_2; Y_{2,i}| W_{1},S_{1}^{n}, S_{2}^{n}, U_i)\nonumber\\
&&- I(Y_{1,i} ;Y_{2}^{i-1}|W_2, W_1, Y_{1,i+1}^{n})\nonumber\\
&\leq&                 I(W_1, U_i; Y_{1,i})\nonumber\\
&&+ I(W_2;Y_{2,i}|W_{1},S_{1}^{n}, S_{2}^{n}, U_i),
\end{eqnarray}
where~$(h)$ is since that~$W_2$ is independent of~$W_1$ and~$(i)$ is derived by substituting~$U_{i}=(Y_{2}^{i-1}, Y_{1, i+1}^{n})$.
Similarly, we have
\begin{eqnarray}{}\label{eqn103}
&&n(R_1+R_2)\nonumber\\
&\leq& I(W_{1}; Y_{1}^{n}|W_2)+I(W_{2}; Y_{2}^{n}| S_{1}^{n}, S_{2}^{n}) \nonumber\\
&\stackrel{(j)}{\leq}&\sum_{i=1}^{n}  I(W_{1}; Y_{1, i}| U_{i}, W_2)+ I(W_2, U_{i}; Y_{2,i}| S_{1}^{n}, S_{2}^{n}),\nonumber\\
&&
\end{eqnarray}
and~$(j)$ is derived by using the same approach as~\eqref{eqn102}. 
Finally, for the equivocation-rate region~$R_{e_{2}}$, we derive
the outer bound, using the approach taken in~\cite{bibi31}, as follows:
\begin{eqnarray}{}\label{eqn104}
n R_{e_2}&\leq& H(W_{2}| Y_{1}^{n})\nonumber\\
&=&     H(W_{2})-I(W_{2};Y_{1}^{n})\nonumber\\
&=&     I(W_{2};Y_{2}^{n})-I(W_{2};Y_{1}^{n})+H(W_{2}| Y_{2}^{n})\nonumber\\
&\stackrel{(k)}{\leq}& I(W_{2};Y_{2}^{n})- I(W_{2};Y_{1}^{n})+2n\epsilon_{n},
\end{eqnarray}
where~$(k)$ is derived from the \textit{Channel Coding Theorem}~\cite[Chapter 7]{bibi34} which implies that
in a reliable communication, the entropy of~$W_2$ given~$Y_{2}^{n}$ is
less than~$\epsilon_{n}$ which is negligible as~$n\rightarrow\infty$.
Then, we have
\begin{eqnarray}{}\label{eqn106}
I(W_{2};Y_{1}^{n}) &=& \sum_{i=1}^{n}I(W_{2};Y_{1,i}    \mid Y_{1, i+1}^{n})\nonumber\\
&=& \sum_{i=1}^{n}I(W_{2};Y_{1,i}    \mid Y_2^{i-1}, Y_{1, i+1}^{n})\nonumber\\
&&+ I(Y_{2}^{i-1};Y_{1,i}\mid Y_{1, i+1}^{n})\nonumber\\
&&-I(Y_{2}^{i-1};Y_{1, i}\mid Y_{1, i+1}^{n},W_{2}).
\end{eqnarray}
Therefore, for~\eqref{eqn104} we have
\begin{eqnarray}{}\label{eqn108}
nR_{e_2}&\leq&\sum_{i=1}^{n}\Big[I(W_{2},Y_{1, i+1}^n; Y_{2,i}\mid Y_{2}^{i-1})\nonumber\\
&&- I(Y_{1, i+1}^n; Y_{2,i}\mid Y_{2}^{i-1}, W_2)\Big]\nonumber\\
&&     -\sum_{i=1}^{n}I(W_{2};Y_{1,i}\mid Y_{2}^{i-1},Y_{1, i+1}^{n})\nonumber\\
&&- I(Y_{2}^{i-1}; Y_{1,i}\mid Y_{1, i+1}^{n})\nonumber\\
&& + I(Y_{2}^{i-1}; Y_{1,i}\mid Y_{1, i+1}^{n}, W_2)\nonumber\\
&\stackrel{(l)}{\leq}& \sum_{i=1}^{n} I(W_{2}; Y_{2,i}\mid Y_{2}^{i-1}, Y_{1, i+1}^n) \nonumber\\
&&- I(W_2;  Y_{1,i}\mid Y_{2}^{i-1}, Y_{1,i+1}^n),
\end{eqnarray}
where~$(l)$ is derived from the Csisz\'{a}r sum identity~\cite{bibi69} which implies that
$\sum_{i=1}^{n}I(Y_{2}^{i-1}; Y_{1,i}\mid Y_{1, i+1}^{n}, W_2) = \sum_{i=1}^{n} I(Y_{1, i+1}^{n}; Y_{2,i}\mid Y_{2, i-1}^{n}, W_2)$,
and the non-negativity of the mutual information function. Similarly, it can be shown that
\begin{eqnarray}{}\label{eqn109}
nR_{e_2}&\leq& \sum_{i=1}^{n} I(W_{2}; Y_{2,i}\mid Y_{2}^{i-1}, Y_{1, i+1}^n, W_1) \nonumber\\
&& - I(W_2;  Y_{1,i}\mid Y_{2}^{i-1}, Y_{1,i+1}^n, W_1).
\end{eqnarray}
Now, by substituting~$U_{i}=(Y_{2}^{i-1}, Y_{1, i+1}^{n})$,~$V_{1,i}= W_{1}$,~$V_{2,i}= W_{2}$, the proof is completed.
\end{proof}

\section{Proof of Theorem~\ref{thm4}}\label{C}

\begin{proof}[Proof of Theorem~\ref{thm4}]
For~$R_1$, using Fano's inequality we have
\begin{eqnarray}{}\label{eqn110}
nR_{1} &\leq& I(W_{1}; Y_{1}^{n})\nonumber\\
&\leq& I(W_{1}, W_2; Y_{1}^{n})\nonumber\\
&\leq& H(Y_{1}^{n})- H(Y_{1}^{n}| W_1, W_2)\nonumber\\
&\leq& H(Y_{1}^{n})- H(Y_{1}^{n}| W_1, W_2, X_{1}^{n}, X_{2}^{n})\nonumber\\
&\stackrel{(m)}{\leq}& H(Y_{1}^{n})- H(Y_{1}^{n}| X_{1}^{n}, X_{2}^{n})\nonumber\\
&\leq& I(X_{1}^{n}, X_{2}^{n}; Y_{1}^{n}),
\end{eqnarray}
where~$(m)$ is due to the fact that~$Y^n_1$ is independent of~$(W_1,W_2)$ given~$(X_1^n ,X_2^n )$. Now, from Fano's inequality
we have
\begin{eqnarray}{}\label{eqn111}
nR_{2}
&\leq& I(W_{2}; Y_{2}^{n}| S_1^n, S_2^n)\nonumber\\
&\leq& I(W_2; Y_{2}^{n}|W_1, S_1^n, S_2^n)\nonumber\\
&\leq& I(W_2, X_2^n; Y_{2}^{n}|W_1, X_1^n (W_1), S_1^n, S_2^n)\nonumber\\
&=&    H(Y_{2}^{n}|W_1, X_1^n, S_1^n, S_2^n)\nonumber\\
&&- H(Y_{2}^{n}|W_1, W_2, X_1^n, X_2^n, S_1^n, S_2^n)\nonumber\\
&\stackrel{(n)}{\leq}& H(Y_{2}^{n}| X_1^n, S_1^n, S_2^n)\nonumber\\
&&- H(Y_{2}^{n}| X_1^n, X_2^n, S_1^n, S_2^n)\nonumber\\
&=& I(X_{2}^{n}; Y_{2}^{n} | X_1^n, S_1^n, S_2^n),
\end{eqnarray}
where~$(n)$ is because of the fact that conditioning does not increase the entropy function and~$Y^n_2$ is independent
of~$(W_1,W_2)$ given~$(X^n_1 ,X^n_2 , S^n_1 , S^n_2 )$.
Now, for~$R_1+ R_2$, from Fano's inequality we have
\begin{eqnarray}{}\label{eqn112}
&&n(R_{1}+R_2)\nonumber\\
&\leq& I(W_2; Y_2^{n}| S_1^n, S_2^n)+ I(W_1; Y_1^n) \nonumber\\
&\leq& I(W_2; Y_2^{n}| W_1, S_1^n, S_2^n)+ I(W_1; Y_1^n) \nonumber\\
&\leq& I(W_2; Y_2^{n}, Y_1^{'n}| W_1, S_1^n, S_2^n)+ I(W_1; Y_1^{n}) \nonumber\\
&\leq& I(W_2; Y_1^{'n}| W_1, S_1^n, S_2^n)\nonumber\\
&&+ I(W_2; Y_2^{n}| Y_1^{'n}, W_1, S_1^n, S_2^n)+ I(W_1; Y_1^{n}) \nonumber\\
&=& H(Y_1^{'n}|W_1, S_1^n, S_2^n)- H(Y_1^{'n}|W_1, W_2, S_1^n, S_2^n)\nonumber\\
&&+ H(Y_2^{n}| Y_1^{'n}, W_1, S_1^n, S_2^n) \nonumber\\
&&- H(Y_2^{n}| Y_1^{'n}, W_1, W_2, S_1^n, S_2^n) + H(Y_1^n)\nonumber\\
&&- H(Y_1^n|W_1) \nonumber\\
&\stackrel{(o)}{\leq}& -H(Y_1^{'n}| W_1, W_2, S_1^n, S_2^n)\nonumber\\
&&+ H(Y_2^{n}| Y_1^{'n}, W_1, S_1^n, S_2^n) \nonumber\\
&&    -H(Y_2^{n}| Y_1^{'n}, W_1, W_2, S_1^n, S_2^n) + H(Y_1^n) \nonumber\\
&\stackrel{(p)}{\leq}&-H(Y_1^{'n}| W_1, W_2, X_1, X_2, S_1^n, S_2^n) \nonumber\\
&&+    H(Y_2^{n}| Y_1^{'n}, W_1,X_1, S_1^n, S_2^n) \nonumber\\
&&-    H(Y_2^{n}| Y_1^{'n}, W_1, W_2, X_1, X_2, S_1^n, S_2^n) + H(Y_1^n)\nonumber
\end{eqnarray}
\begin{eqnarray}
&\leq&-H(Y_1^{'n}| X_1^n, X_2^n, S_1^n, S_2^n)\nonumber\\
&&+ H(Y_2^{n}| Y_1^{'n}, X_1^n, S_1^n, S_2^n) \nonumber\\
&&    -H(Y_2^{n}| Y_1^{'n}, X_1^n, X_2^n, S_1^n, S_2^n) + H(Y_1^n) \nonumber\\
&\leq& I(Y_1^{n}; X_1^n, X_2^n, S_1^n, S_2^n)\nonumber\\
&&+ I(Y_2^{n} ; X_2^n| Y_1^{'n}, X_1^n, S_1^n, S_2^n),
\end{eqnarray}
where~$(o)$ is because that~$H(Y_1^{'n}| W_1, S_1^n, S_2^n)- H(Y_1^{n}| W_1) \leq 0$, and~$(p)$ is due to the fact the conditioning
does not increase the entropy. Thus, the proof is completed.
\end{proof}

\section{Proof of Theorem~\ref{thm7}}\label{D}

\begin{proof}[Proof of Theorem~\ref{thm7}]
To derive the equivocation-rate region~\eqref{eqn44}--\eqref{eqn49}, first we should propose the code-book generation and
the encoding-decoding schemes.

\textbf{Code-book generation:}
\begin{enumerate}
\item
Generate~$2^{nR_1}$ codewords~$x^n_1 (w_1), w_1 \in\{1, 2, \ldots , 2^{nR_1}\}$, choosing~$x^n_1 (w_1)$ independently according to~$P_{X_1} (.)$.
\item
For each~$x^n_1 (w_1)$, generate~$2^{n\tilde{R}_1}$ codewords~$t^n(w_1, v_1)$ using~$\prod_{i=1}^n P_{T|X_1}(.|x^n_1 (w_1))$, where~$v_1 \in
\{1, 2, \ldots 2^{n\tilde{R}_1}\}$.
\item

For each~$x^n_1 (w_1)$ and~$t^n(w_1, v_1)$, generate~$u^n(w_1, v_1,w_{21}, v_{21})$ with i.i.d components based on~$P_{U|X_1 T}$, in
which~$w_{21} \in \{1, 2, \ldots 2^{nR_{21}}\}$ and~$v_{21} \in \{1, 2, \ldots 2^{n\tilde{R}_{21}}\}$.
\item
For each~$x^n_1 (w_1)$,~$t^n(w_1, v_1)$ and~$u^n(w_1, v_1,w_{21}, v_{21})$ generate~$v^n(w_1, v_1, w_{21}, v_{21}, w_{22}, v_{22})$ with i.i.d components based on~$P_{V|X_1 T U}$, in
which~$w_{22} \in \{1, 2, \ldots 2^{nR_{22}}\}$ and~$v_{22} \in \{1, 2, \ldots 2^{n\tilde{R}_{22}}\}$.
\item
Now, distribute~$v^n$ sequences randomly to~$2^{nR}$ bin such that each bin contains~$2^{nM}$ sequences, where~$
R =R_{22}-M$ and~$M = \max\{I(V ; S_1|U,X_1, T), I(V ; Y_1|U,X_1, T)\}$. Then, index each bin by~$j \in \{1, 2,\ldots , 2^{nR}\}$.
Next, partition~$2^{nM}$ sequences in every bin into~$2^{n[M-I(V ;Y_1|U,X_1,T )]}$ subbin each subbin
contains~$2^{nI(V ;Y_1|U,X_1,T)}$ sequences. Index each subbin by~$a \in \{1, 2, \ldots, $ $2^{n[M-I(V ;Y_1|U,X_1,T )]}\}$ and let~$A$ be the random variable to
represent the index of the subbin, and let~$B$ be the random variable to represent the index of the sequences in
each subbin.
\end{enumerate}

\textbf{Encoding:}
Define~$\mathcal{A} = {1, 2, \ldots ,A}$ and~$B = {1, 2, \ldots,B}$ where~$A$ and~$B$ are defined before.
Let~$\mathcal{W}_{22} =\mathcal{A} \times \mathcal{C}$ where~$\mathcal{C} = \{1, 2, \ldots ,B\}$.
Now, define the mapping~$g : \mathcal{B}\rightarrow \mathcal{C}$ to map~$\mathcal{B}$ into~$\mathcal{C}$ subsets with nearly equal
size. Encoder~1 for given~$w_1$, transmits~$x^n_1 (w_1)$.
Encoder~2 for given~$w_1, x^n_1 (w_1)$ and~$s^n_1$, chooses~$t^n(w_1, v_1)$ such
that~$(t^n(w_1, v_1), x^n_1 (w_1), s^n_1 )\in T^{(n)}_{\epsilon}(P_{T,X_1,S_1} )$.
For given~$w_{21}$ and~$t^n(w_1, v_1)$ it chooses~$u^n(w_1, v_1,w_{21}, v_{21})$ such
that~$(u^n, t^n, x^n_1 , s^n_1 )\in T^{(n)}_{\epsilon} (P_{U, T, X_1, S_1})$.
Next, for given~$w_{22}$, it uses the mapping~$w_{22} = (a, c)\rightarrow (a, b)$ which~$b$ is chosen randomly from
the set~$g^{-1}(c)\subset \mathcal{B}$. Then, it chooses~$v^n(w_1, v_1, w_{21}, v_{21},w_{22}(a, b), v_{22})$
such that~$(v^n, u^n, t^n, x^n_1 , s^n_1 ) \in T^{(n)}_{\epsilon}(P_{V, U, T, X_1, S_1})$. Finally, it transmits~$x^n_2 (v^n, u^n, t^n, x^n_1 , s^n_1)$.

\textbf{Decoding:}
Decoder~1, given~$y^n_1$, finds~$(\hat{w}_1, \hat{v}_1, \hat{w}_{21}, \hat{v}_{21})$ such that
$(u^n(\hat{w}_1, \hat{v}_1, \hat{w}_{21}, \hat{v}_{21}), t^n(\hat{w}_1,
\hat{v}_1), x^n_1(\hat{w}_1) , s^n_2 )\in T^{(n)}_{\epsilon}(P_{U, T, X_1, S_2})$.
Decoder~2, given~$y^n_2$ and~$s^n_2$, finds~$(\hat{w}_1, \hat{v}_1, \hat{w}_{21}, \hat{v}_{21}, \hat{w}_{22}(a, b), \hat{v}_{22})$ such that
$(v^n(\hat{w}_1, \hat{v}_1, \hat{w}_{21}, \hat{v}_{21}, \hat{w}_{22}(a, b), \hat{v}_{22}), u^n(\hat{w}_1, \hat{v}_1, \hat{w}_{21}, \hat{v}_{21}),$ $ t^n(\hat{w}_1, \hat{v}_1) , x^n_1(\hat{w}_1) , s^n_2 )
\in T^{(n)}_{\epsilon}(P_{V, U, T, X_1, S_2
})$.

\textbf{Error analysis:} First, fix the channel joint distribution as~\eqref{eqn57}.
The error analysis is similar to the one presented in~\cite{bibi82}.
Thus, the equations~\eqref{eqn44}--\eqref{eqn48} are derived by combining these results.

\textbf{Equivocation-rate calculation:} The equivocation of the~$W_2$ at receiver~1 is calculated as follows:
\begin{eqnarray}{}\label{eqn113}
&& H(W_2|Y^n_1) \nonumber\\
&=& H(W_2, Y^n_1) - H(Y^n_1)\nonumber\\
&=& H(W_2, Y^n_1, A, W_1) - H(Y^n_1)- H(A, W_1| W_2, Y^n_1)\nonumber\\
&=& H(W_2, A, W_1, Y^n_1, V^n) - H(V^n| W_2, A, W_1, Y^n_1)\nonumber\\
&&- H(Y^n_1)- H(A, W_1| W_2, Y^n_1)\nonumber\\
&=& H(W_2, A, W_1| Y^n_1, V^n) + H(Y^n_1, V^n)- H(Y^n_1)\nonumber\\
&&- H(V^n| W_2, A, W_1, Y^n_1)- H(A, W_1| W_2, Y^n_1)\nonumber\\
&\stackrel{(q)}{\geq}& H(V^n| Y^n_1)- H(V^n| W_2, A, W_1, Y^n_1)\nonumber\\
&&- H(A, W_1| W_2, Y^n_1)\nonumber\\
&\stackrel{(r)}{\geq}& H(V^n| Y^n_1, U^n, X_1^n, T^n)- H(V^n| W_2, A, W_1, Y^n_1)\nonumber\\
&&-\log|\mathcal{A}|- H(V^n| Y^n_2, S_2^n, U^n, X_1^n, T^n)\nonumber\\
&\stackrel{(s)}{\geq}& n\Big[I(V; Y_2, S_2|U, X_1, T)-I(V; Y_1|U, X_1, T)\Big]\nonumber\\
&&- H(V^n| W_2, A, W_1, Y^n_1)\nonumber\\
&& - \Big[\max\{I(V; S_1|U, X_1, T), I(V; Y_1|U, X_1, T)\}\nonumber\\
&&- I(V; Y_1| U, X_1, T)\Big]\nonumber\\
&\geq& n\Big[I(V; Y_2, S_2|U, X_1, T)\nonumber\\
&&- \max\{I(V; S_1|U, X_1, T), I(V; Y_1|U, X_1, T)\}\Big],
\end{eqnarray}
where~$(q)$ follows from the non-negativity of entropy function;~$(r)$ follows from the fact that conditioning does
not increase the entropy, the non-negativity of entropy function, and the fact that
$H(A, W_1| W_2, Y^n_1)= H(A| W_2, Y^n_1)+ H(W_1|A, W_2, Y^n_1) \leq H(A)\leq \log|\mathcal{A}|$, thanks to~$H(W_1|A, W_2, Y^n_1)=0$;
$(s)$ is because of Fano's inequality which implies that the term~$H(V^n| W_2, A, W_1, Y^n_1)$ tends to zero for~$n\rightarrow\infty$
(see~\cite{bibi70}). The proof is completed.
\end{proof}

\ack
This work was partially supported by Iran National Science Foundation (INSF), under contracts' numbers 91/s/26278 and 92/32575.
Parts of this paper was presented in the International Symposium on Information Theory and Application (ISITA) 2010.

The authors gratefully acknowledge the anonymous reviewers for their suggestions,
comments and corrections as well as those of the associate editor.


\end{document}